\def\notes{0}

\documentclass[11pt]{article}
\usepackage{hyperref}
\usepackage[utf8]{inputenc}

\usepackage[T1]{fontenc}
\usepackage{lmodern}
\usepackage{amssymb,amsmath,amsthm}
\usepackage{amsbsy}
\usepackage{bm}
\usepackage{ifxetex,ifluatex}
\usepackage{bbm}
\usepackage{todonotes}
\usepackage{tikz}
\usepackage{xifthen}
\usepackage{thmtools,thm-restate}

\usepackage[normalem]{ulem}
\usepackage{color}

\usepackage{graphicx}
\usepackage{subcaption}
\usepackage{listings}
\usepackage{multicol}
\usepackage{etoolbox}
\usepackage{wrapfig}
\usepackage{xcolor}
\usepackage{thmtools}
\usepackage{thm-restate}
\usepackage[margin=1in]{geometry}
\usepackage[shortlabels]{enumitem}
\setlist{topsep=1pt,parsep=1pt,partopsep=1pt,itemsep=1pt}

\newcommand{\1}{\mathbbm{1}}
\DeclareMathOperator*{\argmin}{argmin}
\DeclareMathOperator*{\argmax}{argmax}
\DeclareMathOperator*{\poly}{poly}
\DeclareMathOperator*{\E}{\mathbb{E}}

\newcommand{\F}{\mathbb{F}}
\renewcommand{\P}{\Pr}
\renewcommand{\bar}{\overline}
\renewcommand{\epsilon}{\varepsilon}
\newcommand{\eps}{\varepsilon}

\newcommand{\cF}{\mathcal{F}}
\newcommand{\C}{\mathcal{C}}

\newcommand{\Oh}{\mathcal{O}}
\newcommand{\bH}{\bar{H}}

\newcommand{\usef}{\sqsubset_u}
\newcommand{\dom}{\preceq}
\newcommand{\supp}{\mathrm{supp}}

\newcommand{\wt}{\mathrm{wt}}

\newcommand{\beq}{\begin{equation}}
\newcommand{\eeq}{\end{equation}}
\def\bal#1\eal{\begin{align}#1\end{align}}

\newtheorem{theorem}{Theorem}[section]
\newtheorem{lemma}[theorem]{Lemma}

\newtheorem{corollary}[theorem]{Corollary}
\newtheorem{definition}[theorem]{Definition}

\newtheorem{claim}[theorem]{Claim}

\newcommand{\HadMat}{\left[\begin{array}{cc} 1 & 0 \\ \alpha & 1 \end{array}\right]}

\newcommand{\bvec}[1]{ {\boldsymbol{#1}} }
\let\Oldforall\forall
\renewcommand{\forall}{~\Oldforall} % add some space before foralls

\let\Oldinf\inf
\renewcommand{\inf}{\Oldinf\limits}
\let\Oldsup\sup
\renewcommand{\sup}{\Oldsup\limits}

\ifnum\notes=1
\newcommand{\arnote}[2][]{%
\ifthenelse{\equal{#1}{}}%
{\todo[inline,color=green]{{\small #2} --\textsc{Atri}}}%
{
}
}

\newcommand{\mnote}[2][]{%
\ifthenelse{\equal{#1}{}}%
{\todo[inline,color=orange]{{\small #2} --\textsc{Madhu}}}%
{
}
}

\newcommand{\vnote}[2][]{%
\ifthenelse{\equal{#1}{}}%
{\todo[inline,color=pink]{{\small #2} --\textsc{Venkat}}}%
{
}
}

\newcommand{\pnote}[2][]{%
\ifthenelse{\equal{#1}{}}%
{\todo[inline,color=cyan]{{\small #2} --\textsc{Preetum}}}%
{
}
}

\newcommand{\jnote}[2][]{%
\ifthenelse{\equal{#1}{}}%
{\todo[inline,color=yellow]{{\small #2} --\textsc{Jarek}}}%
{
}
}
\else

\newcommand{\arnote}[1]{%
{
}
}

\newcommand{\mnote}[1]{%
{
}
}

\newcommand{\vnote}[1]{%
{
}
}

\newcommand{\pnote}[1]{%
{
}
}

\newcommand{\jnote}[1]{%
{
}
}
\fi

\newcommand{\mtodo}[1]{
\ifnum\notes=1
\todo{#1}
\fi}

\newcommand{\Arikan}{Ar\i kan}

\newcommand{\Dec}{\mathrm{Dec}}
\renewcommand{\H}{\bar{H}}
\usepackage{algorithm}
\usepackage[noend]{algpseudocode}
\algrenewcommand\algorithmicrequire{\textbf{Input:}}
\algrenewcommand\algorithmicensure{\textbf{Output:}}
\algrenewcommand\algorithmicwhile{\textbf{While}}
\algrenewcommand\algorithmicfor{\textbf{For}}
\algrenewcommand\algorithmicreturn{\textbf{Return}}
\algrenewcommand\algorithmicif{\textbf{If}}

\algnewcommand\algorithmicconst{\textbf{Constants:}}
\algnewcommand\Const{\item[\algorithmicconst]}

\title{\textbf{Polar Codes with exponentially small error at finite block length}}
\author{%
Jaros\l aw B\l asiok\thanks{Harvard John A. Paulson School of Engineering and Applied Sciences, Harvard University, 33 Oxford Street,
Cambridge, MA 02138, USA. Email: {\tt jblasiok@g.harvard.edu.} Supported by ONR grant N00014-15-1-2388.}
\and Venkatesan Guruswami\thanks{Computer Science Department, Carnegie Mellon University, Pittsburgh, PA 15213, USA. This work was done when the author was visiting the Center of Mathematical Sciences and Applications, Harvard University, Cambridge, MA 02138.  {\tt venkatg@cs.cmu.edu}. Research supported in part by NSF grants CCF-1422045 and CCF-1563742.}
\and Madhu Sudan\thanks{Harvard John A. Paulson School of Engineering and Applied Sciences, Harvard University, 33 Oxford Street, Cambridge, MA 02138, USA. Email: {\tt madhu@cs.harvard.edu}. Work supported in part by a Simons Investigator Award and NSF Awards CCF 1565641 and CCF 1715187.}
}
\date{}

\begin{document}
\maketitle

%\mnote{We should think of a more catchy title.}

\thispagestyle{empty}

\begin{abstract}
We show that the entire class of polar codes (up to a natural necessary condition) 
converge to capacity at block lengths polynomial in the gap to capacity, while {\em simultaneously} achieving failure probabilities that are exponentially small in the block length (i.e., decoding fails with probability $\exp(-N^{\Omega(1)})$ for codes of length $N$). Previously this combination was known only for one specific family within the class of polar codes, whereas we establish this whenever the polar code exhibits a condition necessary for any polarization. 

Our results adapt and strengthen a local analysis of polar codes due to the authors with Nakkiran and Rudra [Proc. STOC 2018]. Their analysis related the time-local behavior of a martingale to its global convergence, and this allowed them to prove that the broad class of polar codes converge to capacity at polynomial block lengths. Their analysis easily adapts to show exponentially small failure probabilities, provided the associated martingale, the ``\Arikan\ martingale'', exhibits a corresponding strong local effect. The main contribution of this work is a much stronger local analysis of the \Arikan\ martingale. This leads to the general result claimed above.

In addition to our general result, we also show, for the first time, polar codes that achieve failure probability $\exp(-N^{\beta})$ for any $\beta < 1$ while converging to capacity at block length polynomial in the gap to capacity. Finally we also show that the ``local'' approach can be combined with \emph{any} analysis of failure probability of an arbitrary polar code to get essentially the same failure probability while achieving block length polynomial in the gap to capacity.
\end{abstract}

\newpage

\section{Introduction}

Ever since their discovery ~\cite{arikan-polar} polar codes have been a subject of vast interest, both for their theoretical and practical significance. Theoretical interest in them arises from two desirable features that they exhibit: (1) They give codes of length $N$ (for infinitely many $N$) along with efficient decoding algorithms that correct channel errors with all but exponentially (i.e., $\exp(-N^{\Omega(1)})$) small failure probability. (2) They also converge to capacity extremely fast --- i.e., at block length $N$ which is only polynomial in the inverse of the "gap to capacity". The former effect is known to hold in general, i.e., for the entire class of polar codes (up to a minimal and natural necessary condition). The latter was shown to hold in the same generality only recently~\cite{BGNRS} --- previous works~\cite{GX-focs13,HAU14,GV15} were only able to establish it for one specific construction of polar codes. And while the early works were able to show effects (1) and (2) simultaneously for this construction, the other polar codes were not known to have both features simultaneously.

The main goal of this paper is to remedy this weakness. We show roughly that the techniques of \cite{BGNRS} can be strengthened to achieve both effects simultaneously for the entire broad class of polar codes. In addition to the generality of the result this also leads to quantitative improvements on the error-exponent at polynomially small block lengths in the gap to capacity. We elaborate on these further after some background.

\subsection{Background}

\newcommand{\err}{\mathrm{err}}

In the theory of Shannon, a memoryless channel is given by a probabilistic map  from an input alphabet (a finite field $\F_q$ in this paper) to an output alphabet (an abstract set $\mathcal{Y}$ here). A family of codes $C_N :\F_q^{k_N} \to \F_q^N$ along with decoding algorithm $D_N:\mathcal{Y}^N \to \F_q^{k_N}$ achieves {\em rate} $R$ if $\lim_{N\to \infty} k_N/N \geq R$. It is said to achieve failure probability $\err(N)$ if $\Pr_{M \in \F_q^{k_N}}[D_N(C_N(M)\ne M] \leq \err(N)$ for every $N$. Shannon's celebrated theorem associates a capacity $C$ with every channel such that transmission at rate higher than capacity will have constant failure probability, whereas for every $R < C$, for every sufficiently large $N$, there exist codes of rate $R$ with failure probability $\exp(-\Omega(N))$. 
The quantity $\epsilon \triangleq C-R$ is called the ``gap to capacity''. The relationship between the block length $N$, the gap to capacity $\epsilon$ and the failure probability $\err(N)$ are the central quantities of interest to this paper.

\jnote{Isn't even success probability sub-constant if we try to transmit with rate above capacity?}

The specific family of codes we consider in this paper are ``polar codes'' introduced by \Arikan~\cite{arikan-polar}. These codes are a broad class of (infinite families of) codes, one family for every matrix $M \in \F_q^{k \times k}$ and symmetric channel. The $t$-th code in the sequence has length $k^t$, and is given by (affine shift) of some subset of rows of $M^{\otimes t}$. It is well known that under a simple necessary condition on $M$ (that we call mixing), these codes achieve exponentially small failure probability in a weak sense: Specifically for every symmetric channel, for every mixing $M$, there exists a $\beta > 0$ such that for every $\epsilon > 0$ there exists a $N_0$ such that every code in the family of length $N \geq N_0$ has at most $\epsilon$ gap to capacity and achieves failure probability at most $\exp(-N^\beta)$. 
Indeed by picking $M$ carefully one could achieve $\beta$ arbitrarily close to $1$ (though this approach can not yield $\beta =1$), and moreover for a given matrix $M$, the range of achievable $\beta$ can be explicitly computed from simple combinatorial properties of this matrix \cite{KSU10}. However note that these analyses did not provide explicit relationship between $\epsilon$ and $N_0$. 
\jnote{Exponential dependence of $N_0$ on $\epsilon$ was known?}

It was more recently shown \cite{GX-focs13,HAU14,GV15} that there exists an $M$ (specifically $M = 
\left[\begin{array}{cc} 1 & 0 \\ 1 & 1\end{array} \right]$)
such that the associated code achieves exponentially small failure probability even at polynomially small block lengths --- i.e., when $N_0 = \poly(1/\epsilon)$. The $\beta$ associated with this result is bounded well away from $1$. But till last year no other code (for any other matrix $M$) was even known to achieve failure probability going to zero for polynomially small block lengths. This was remedied in part by a previous work of the authors with Nakkiran and Rudra~\cite{BGNRS} where they showed that for every mixing matrix $M$ and every symmetric channel the associated code converges at block length growing polynomially with gap to capacity, however their failure probability analysis only yielded $\err(N) \leq 1/\poly(N)$. Their work forms the starting point of this work.
% \jnote{ More discussion about [6]. }

\subsection{Our results}

Our results show that it is possible to combine the general analyses for ``polynomial convergence of block length in gap to capacity'' (from \cite{BGNRS}) with any strong analysis of the failure probability. Specifically we show the following:
\begin{enumerate}
\item For every mixing matrix $M$ and symmetric channel the associated family of polar codes yield exponentially small decoding failure at block lengths polynomial in the gap to capacity.
\item  While the result in Part (1) is general the resulting $\beta$ may not be optimal. We complement this with a result showing that for every $\beta < 1$ there exist polar codes associated with some matrix $M$, that get close to capacity at polynomial block length with decoding failure probability being $\exp(-N^{\beta})$. We note that no previous analysis yielded such quantitatively strong bounds on any family of polar codes with polynomial block length.
\item Finally we show that convergence to capacity at polynomial block length comes with almost no price in the failure probability. We show this by proving that if any polar code achieves capacity (even if at very large block lengths) with failure probability $\exp(-N^\beta)$, then for every $\beta' < \beta$ it achieves capacity with failure probability $\exp(-N^{\beta'})$
where the block length is a polynomial $p_{\beta,\beta'}(1/\epsilon)$.
\end{enumerate}

%While the third result subsumes the previous two (when combined with known results in the literature), the first two have relatively simple and self-contained (after including the work of \cite{BGNRS}) proofs.
While the third result subsumes the previous two (when combined with known results in the literature), we include the first two to show that it is possible to prove strong results about failure probabilities $\exp(-N^\beta)$ with blocklength polynomial in the gap to capacity, entirely within the local polarization framework developed in \cite{BGNRS} and here --- without appealing to previous analyses. In fact the proofs of those two are quite simple (given the work of \cite{BGNRS}).

% \jnote{In fact the third result itself also have not-too-difficult proof. I guess the point of including (1) and (2) that it is possible to  show exponential polarization from scratch completely within our framework, without appealing at all to the Bhattacharaya analysis, and previous works. The sentence above doesn't completely capture this motivation. }

On the other hand, for given matrix $M$, the optimal exponent $\beta$ was exactly characterized in terms of explicit combinatorial properties of matrix $M$ --- but with potentially very large blocklengths \cite{KSU10}. The third result of our paper automatically lifts this theorem to the setting where blocklength is polynomial in the gap to capacity --- given matrix $M$ one can compute the ``correct'' exponent $\beta$ as in \cite{KSU10}, and essentially the same exponent is achievable already within polynomial blocklength, whereas no larger exponent is achievable, regardless of how large blocklength one takes.

\subsection{Techniques}

We now turn to the central ingredient in our analyses of polar codes which we inherit from \cite{BGNRS}, namely the ``local'' analysis of $[0,1]$-martingales. It is well-known that the analysis of polar codes can be tied to the analysis of an associated martingale, called the \Arikan\ martingale in \cite{BGNRS}. Specifically given a channel and a matrix $M$ one can design a martingale $X_0,X_1,\ldots,X_t,\ldots$ with $X_t \in [0,1]$, such that the performance of the code of length $k^t$ depends on the behavior of the random variable $X_t$. Specifically to achieve $\epsilon$ gap to capacity with failure probability $\rho=\err(N)$, the associated martingale should satisfy $\Pr[X_t \in (\rho/N,1-\eps/2)] \leq \eps/2$. Considering the fact that we want the failure to be exponentially small in $N$ and $\eps$ to be inverse polynomially small in $N$ and noting $N = k^t$, this requires us to prove that $\Pr[X_t \in (\exp(-\exp(O(t))),1 - \exp(-\Omega(t))]\leq \exp(-\Omega(t))$. 

Usual proofs of this property typically track many aspects of the distribution of $X_t$, whereas a ``local'' analysis simply reasons about the distribution of $X_t$ conditioned on $X_{t-1}$. For the \Arikan\ martingale (as for many other natural martingales) this one-step evolution is much easier to describe than the cumulative effects of $t$-steps. In \cite{BGNRS} a simple local property, called ``local polarization'', of this one-step evolution was described (enforcing that the random variable has enough variance if it is not close to the boundary $\{0,1\}$ and that it gets sucked to the boundary when it is close). It was then shown that local polarization leads to global polarization, though only for $\rho = 1/\poly(N)$ --- specifically they showed that $\Pr[X_t \in (1/\poly(N),1-\eps/2)] \leq \eps/2$.

It is easy to modify the definition of local polarization slightly to get a stronger definition that would imply the desired convergence even for $\rho(N) = \exp(-N^{\Omega(1)})$. Indeed we do so, calling it ``exponential local polarization'' of a martingale, and show that this stronger local polarization leads to exponentially small failure probabilities.

The crux of this paper is in showing that the \Arikan\ martingale exhibits exponential local polarization. For readers familiar with the technical aspects, this might even be surprising. In fact the most well-studied \Arikan\ martingale, the one associated with the binary symmetric channel and the matrix $M=\left[\begin{array}{cc} 1 & 0 \\ 1 & 1\end{array} \right]$ is not exponentially locally polarizing. We get around this seemingly forbidding barrier by showing that the martingale associated with
$M^{\otimes 2}$ (the tensor-product of $M$ with itself) is exponentially locally polarizing, and this is almost as good for us. (Instead of reasoning about the martingale $X_0,X_1,X_2,\ldots,$ this allows us to reason about $X_0,X_2,X_4,\ldots$ which is sufficient for us.)
Combined with some general reductions as in \cite{BGNRS} this allows us to show that for every symmetric channel and every mixing matrix, the associated martingale is exponentially locally polarizing and this yields our first main result above.

To get failure probability $\exp(-N^\beta)$ for $\beta \to 1$ we show that if the matrix $M$ contains the parity check matrix of a code of sufficiently high distance then the \Arikan\ martingale associated with $M$ exhibits exponential local polarization over any symmetric channel, and in turn this leads to codes whose failure probability is $\exp(-N^\beta)$ for $\beta \to 1$. 

Finally we turn to our last result showing that any matrix producing codes with failure probability $\exp(-N^\beta)$ (but not necessarily for $N = \poly(1/\eps)$) also gets failure probability $\exp(-N^{\beta'})$ for $N \geq p_{\beta,\beta'}(1/\epsilon)$ for some polynomial $p_{\beta,\beta'}$, and any $\beta' < \beta$. This result is obtained by showing that if $M$ achieves exponentially small error, then for some large $t_0 = t_0(\beta,\beta')$, the matrix $M^{\otimes t_0}$ contains the parity check matrix of a high-distance code, with distance high enough to imply that its failure probability is $\exp(-N^{\beta'})$.

%\jnote{Is it good place to mention lifting from additive channels to symmetric channels?}

\section{Main Definitions and Results}

\subsection{Martingales and Polarization}

In this section we let $X_0,X_1,X_2,\ldots$ be a $[0,1]$-bounded martingale, i.e., $X_t \in [0,1]$ for all $t$ and for every $x_0,\ldots,x_t$, $\E[X_{t+1}| X_0 = x_0, \cdots, X_t = x_t] = x_t$.

We say that a martingale has exponentially strong polarization if the probability that
$X_t$ is not close (as a function of $t$) to the boundary $\{0,1\}$ is exponentially small in $t$. Formally
\begin{definition}[Exponentially Strong Polarization]
	We say that $X_t$ has $\Lambda$-exponentially strong polarization if for every $0 < \gamma < 1$ there exist constants $\alpha < \infty$ and $0< \rho < 1$ such that for every $t$, $\Pr[X_t \in ( 2^{-2^{\Lambda \cdot t}}, 1 - \gamma^t)] \leq \alpha\cdot\rho^t$. 
\end{definition}

Note that this definition is asymmetric --- paths of the martingale that converge to zero, have doubly-exponential rate of convergence, whereas those converging to $1$ are doing it only exponentially fast.\footnote{It turns out that for the polar coding application, the behavior of the martingale at the lower end is important as it governs the decoding error probability, whereas behavior of the martingale near the upper end is not that important. The probability that the martingale doesn't polarize corresponds to the gap to capacity.}
This should be compared with the notion of strong polarization present in \cite{BGNRS}, namely
\begin{definition}[Strong Polarization]
	We say that $X_t$ has strong polarization if for every $0 < \gamma < 1$ there exist constants $\alpha < \infty$ and $0< \rho < 1$ such that for every $t$, $\Pr[X_t \in (\gamma^t, 1 - \gamma^t)] \leq \alpha\cdot\rho^t$. 
\end{definition}

As in \cite{BGNRS} the notion of Exponential Strong Polarization is not a local one but rather depends on the long run behavior of $X_t$.
A notion of local polarization, that only relates the evolution of $X_{t+1}$ from $X_t$, was defined in \cite{BGNRS}, and shown to imply strong polarization. Let us recall this definition.
\begin{definition}[Local Polarization]
\label{defn:polar-local}
 A $[0,1]$-martingale sequence $X_0,\ldots,X_j,\ldots,$ is {\em locally polarizing} if the following conditions hold:
\begin{enumerate}
\item {\bf (Variance in the middle):} For every $\tau > 0$, there is a $\theta = \theta(\tau) > 0$ such that for all $j$, we have: If $X_j \in (\tau, 1-\tau)$ then
	$\E [(X_{j+1} - X_j)^2 | X_j] \geq \theta$.
\item {\bf (Suction at the ends):} There exists an $\alpha > 0$, such that for all  $c < \infty$, there exists a $\tau =\tau(c) > 0, $
such that:
\begin{enumerate}
\item
If $X_j \leq \tau$ then $\Pr[X_{j+1} \leq X_j/c | X_j]\geq \alpha$.
\item
Similarly, if $1 - X_j \leq \tau$ then $\Pr[(1 - X_{j+1} \leq (1 - X_j)/c | X_j] \geq \alpha$.
\end{enumerate}
We refer to condition (a) above as {\em Suction at the low end} and condition (b) as {\em Suction at the high end}.
\end{enumerate}
When we wish to be more explicit, we refer to the sequence as $(\alpha,\tau(\cdot),\theta(\cdot))$-locally polarizing.
\end{definition}

With an eye toward showing exponential strong polarization also via a
local analysis, we now define a concept of local polarization tailored
to exponential polarization.

%The following provides an alternate definition of exponential polarization that only relates $X_{t+1}$ to $X_t$. 
%We introduce additional local property of a martingale sequence.
\begin{definition}[Exponential Local Polarization]
	We say that $X_t$ has $(\eta, b)$-exponential local polarization if it satisfies local polarization, and the following additional property
	\begin{enumerate}
			\item {\bf(Strong suction at the low end):} There exists $\tau > 0$ such that if $X_j \leq \tau$ then
			$\Pr[X_{j+1} \leq X_j^{b} | X_j] \geq \eta$.
	\end{enumerate}
	\end{definition}
In the same way as local polarization implies the strong global polarization of a martingale~\cite[Theorem 1.6]{BGNRS}, this new stronger local condition implies a stronger global polarization behavior.
\begin{theorem}[Local to Global Exponential Polarization]
	\label{thm:local-to-global}
	Let $\Lambda < \eta \log_2 b$.  
Then if a $[0,1]$-bounded martingale $X_0,X_1,X_2,\ldots$ satisfies $(\eta, b)$-exponential local polarization then it also satisfies $\Lambda$-exponentially strong polarization.
\end{theorem}

The proof of this theorem follows the same outline as the proof of Theorem~1.6 in \cite{BGNRS}, and we present it in Section~\ref{sec:local-to-global}.

\subsection{Matrix Polarization}

In this section we relate statements about the local polarization of the \Arikan\ martingale associated with some matrix $M$ (and some channel) to structural properties of $M$ itself. The formal definition of the \Arikan\ martingale is included for completeness in Appendix~\ref{sec:arikan}, but will not be used in this paper.

We first recall the definition of a mixing matrix --- it is a simple necessary condition for associated \Arikan\ martingale to be non-trivial (i.e. non-constant).

\begin{definition}[Mixing matrix]
For prime $q$ and $M \in \F_q^{k \times k}$, $M$ is said to be a mixing matrix if $M$ is invertible and for every permutation of the rows of $M$, the resulting matrix is not upper-triangular.
\end{definition}

Let us now rewrite the (technical) condition of the \Arikan\ martingale
 associated with $M$ being exponentially locally polarizing in more direct terms. This leads us to the following definition.

%\jnote{Term \emph{exponential polarization} suggests symmetry --- but in this paper we care only about extremely strong suction to the lower end, as the actual polarization is already handled in the previous paper.}

\begin{definition}[Exponential polarization of matrix]
		\label{def:matrix-exp-polar}
	We say that a matrix $M \in \F_q^{k \times k}$ satisfies $(\eta, b)$-exponential polarization, if there exist some $\tau > 0$, such that for any $\delta < \tau$ and for any random sequence $(U_1, A_1), \ldots (U_k, A_k)$, where $(U_i, A_i) \in \F_q \time \Sigma$ are i.i.d., and satisfy $\H(U_i | A_i) \leq \delta$, we have 
	\begin{equation*}
		\H( (UM)_j | (UM)_{<j}, A) \leq \delta^{b}
\end{equation*} for at least $\eta$ fraction of indices $j \in [k]$.
\end{definition}

In the above definition and throughout the paper $\H$ refers to normalized entropy, i.e. $\H(X | A) := \frac{1}{\log_2 q} H(X | A)$, so that $\H(X | A) \in [0, 1]$, and $U = (U_1, \ldots U_k)$, similarly $A = (A_1, \ldots A_k)$. Moreover, for a vector $V \in \F^k$, and $j \leq k$, by $V_{<j}$ we denote a vector in $\F^{j-1}$ with coordinates $(V_1, \ldots V_{j-1})$.

The following lemma explicitly asserts that matrix polarization implies martingale polarization (as claimed). 
 
\begin{lemma}\label{lem:matrix-implies-arikan}
	If mixing matrix $M$ satisfies $(\eta, b)$-exponential polarization, then \Arikan~martingale associated with $M$ is $(\eta, b)$-exponentially locally polarizing.
\end{lemma}
The proof of the above lemma is very similar to the proof of Theorem 1.10 in \cite{BGNRS} --- with definitions of \Arikan~martingale and exponential polarization of matrix in hand this proof is routine, although somewhat tedious and notationally heavy. We postpone this proof to the full version of this paper.

In the light of the above, and in context of Theorem~\ref{thm:local-to-global}, we have reduced the problem of showing (global) exponentially strong polarization of \Arikan~martingale, to understanding parameters for exponential polarization of specific matrices, based on the structural propertues of these matrices.

In this paper we provide three results of this form. The first of our results considers mixing matrices and analyzes their local polarization. We recall the definition of a mixing matrix.

It is well known that if a matrix is not mixing then the associated martingale does not polarize at all (and the corresponding martingale satisfies $X_t = X_{t-1}$ for every $t$).
In contrast if the matrix $M$ is mixing, our first lemma shows that $M^{\otimes 2}$ (the tensor-product of $M$ with itself) is exponentially polarizing.

\begin{lemma}
	\label{lem:every-matrix-works}
For every mixing matrix $M \in \F_q^{k \times k}$ and every $\varepsilon > 0$, matrix $M^{\otimes 2}$ satisfies $(\frac{1}{k^2}, 2-\varepsilon)$-exponential polarization.
\end{lemma}

This translates immediately to our first main theorem stated in Section~\ref{sec:codes-implications}.

Our second structural result on matrix polarization shows that matrices that contain the parity check matrix of a high distance code lead to very strong exponential polarization parameters.

\begin{lemma}
	\label{lem:code-exp-suction}
	If a mixing matrix $M$ is decomposed as $M = \left[ M_0 | M_1 \right]$, where $M_0 \in \F_q^{k \times (1 - \eta) k}$ is such that $\ker M_0^T$ is a linear code of distance larger than $2 b$, then matrix $M$ satisfies $(\eta, b - \varepsilon)$-exponential polarization for every $\varepsilon > 0$.
\end{lemma}

By using standard results on existence of codes with good distance, we get as an immediate corollary that there exist matrices with almost optimal exponential polarization parameters.

\begin{corollary}
	\label{cor:beta-close-to-one}
	For every $\varepsilon$ and every prime field $\F_q$, there exist $k$, and matrix $M \in \F_q^{k\times k}$, such that matrix $M$ satisfies $(1 - \varepsilon, k^{1-\varepsilon})$ exponential polarization.
\end{corollary}
\begin{proof}
Consider a parity check matrix $M_0$ of a BCH code with distance $2 k^{1 - \varepsilon}$. We can achieve this with a matrix $M_0 \in \F_q^{k \times k_0}$, where $k_0 = \Oh(k^{1 - \varepsilon} \log k)$. Hence, as soon as $k > \Omega (2^{\varepsilon^{-1} \log \varepsilon^{-1})})$, we have $k_0 < \varepsilon k$. We can now complete $M_0$ to a mixing matrix.
\end{proof}
\jnote{Reviewer asks about a random matrix. It is actually a good idea to say something about it --- should it be just a sidenote, or should we prove it?}
It is worth noting, that by the same argument and standard results on the distance of random linear codes, a random matrix $M \in \F_q^{k\times k}$ with high probability satisfies a $(1-\varepsilon, k^{1-\varepsilon})$ local polarization, with $\varepsilon \to 0$ as $k \to \infty$.

By the whole chain of reductions discussed above, Corollary~\ref{cor:beta-close-to-one} implies that for any $\varepsilon$ there exist polar codes with decoding failure probability $\exp(-N^{1-\varepsilon})$, where the blocklength $N$ depends polynomially in the desired gap to capacity. Moreover, those codes are ubiquitous --- polar codes arising from a large random matrix will usually have this property.

Our final structural result is morally a ``converse'' to the above: It shows that if a matrix $M$ leads to a polar code with exponentially small failure probability then some high tensor power $N = M^{\otimes t}$ of $M$ contains the parity check matrix of a high distance code. In fact more generally if a matrix $P\in \F_q^{k \times s}$ is the parity check matrix of a code which has a decoding algorithm that corrects errors from a $q$-symmetric channel with failure probability $\exp(-k^\beta)$ then this code has high distance. 

\jnote{Is it really ``converse'' of the above?}

\begin{definition}
\label{defn:q-ary-bernoulli}
	For any finite field $\F_q$ we will denote by $B_{q}(\varepsilon)$ the distribution on $\F_q$ such that for $Z \sim B_{q}(\varepsilon)$ we have $\P(Z = 0) = 1 - \varepsilon$, and $\P(Z = k) = \frac{\varepsilon}{q-1}$ for any $k\not=0$.
\end{definition}
\begin{lemma}
\label{lem:compression-implies-distance}
Consider a matrix $P \in \F_q^{k \times s}$ and arbitrary decoding algorithm $\Dec : \F_q^s \to \F_q^k$, such that for independent random variables $U_1, \ldots U_i \sim B_{q}(\varepsilon)$ with $\varepsilon < \frac{1}{2}$, we have $\P(\Dec(UP) \not= U) < \exp(-k^{\gamma})$. Then $\ker P$ is a code of distance at least $k^{\gamma} \log^{-1}(q/\varepsilon)$.
\end{lemma}

This lemma, when combined with Lemma~\ref{lem:code-exp-suction} shows that the only way a polar code associated with a  matrix $M$ can give exponentially small failure probability $\exp(-N^\beta)$ is that some tensor of this matrix is {\em locally exponentially} polarizing and so in particular this matrix also yields exponentially small failure probabilities at block length polynomial in the gap to capacity.

\subsection{Implications for polar codes \label{sec:codes-implications}}

We start this section by including the definition of symmetric channel --- all our results about polar codes show that we can achieve capacity for those channels.
\begin{definition}[Symmetric memoryless channel]
	A $q$-ary \emph{symmetric memoryless channel} is any probabilistic function $\mathcal{C} : \F_q \to \mathcal{Y}$, such that for every $\alpha,\beta \in \F_q$ there is a bijection $\sigma:\mathcal{Y}\to \mathcal{Y}$ such that for every $y \in \mathcal{Y}$ it is the case that $\C_{Y=y|\alpha} = \C_{Y = \sigma(y)|\beta}$, and moreover for any pair $y_1, y_2 \in \mathcal{Y}$, we have $\sum_{x \in \F_q} C_{Y=y_1 | x} = \sum_{x \in \F_q} C_{Y=y_2 | x}$ (see, for example,~\cite[Section 7.2]{CoverThomas}).

	Such probabilistic function yields a probabilistic function $\mathcal{C} : \F_q^N \to \mathcal{Y}^N$, by acting independently on each coordinate.
\end{definition}

We will now recall the following theorem which shows that if the \Arikan\ martingale polarizes then a corresponding code achieves capacity with small failure probability.

\begin{theorem}[Implied by \Arikan~\cite{arikan-polar}]
\label{thm:exp-code}
Let $\C$ be a $q$-ary symmetric memoryless channel and let $M \in \F_q^{k \times k}$ be an
invertible matrix.
If the \Arikan\ martingale associated with $(M,\C)$ is $\Lambda$-exponentially strongly polarizing
then there is a polynomial $p$ such that for every $\epsilon > 0$ and every $N=k^t \geq p(1/\epsilon)$, 
there is a code $C \subseteq \F_q^N$ of dimension at least $(\mathrm{Capacity}(\C)-\epsilon)\cdot n$ such that $C$ is an affine code generated by the restriction of $(M^{-1})^{\otimes t}$ to a subset of its rows and an affine shift. Moreover there is a decoding algorithm for these codes that has failure probability bounded by $\exp(-N^{\Lambda/\log_2 k})$, and running time $\Oh(N \log N)$. The running time of accompanying encoding algorithm is also $\Oh(N \log N)$.
\end{theorem}

We omit the proof of this theorem, which is identical to Theorem~1.7 in \cite{BGNRS} except for minor modifications to incorporate the exponential polarization/failure probability.

Armed with this theorem, we can now convert the structural results asserted in the previous section into convergence and failure probability of polar codes.

\begin{theorem}
\label{thm:thm1}
For every prime $q$, every mixing matrix $M\in\F_q^{k\times k}$, every symmetric memoryless channel $\C$ over $\F_q$, there is a polynomial $p$ and $\beta > 0$ such that for every $\epsilon > 0$ and every $N = k^t \geq p(1/\epsilon)$, 
there is an affine code $C$, that is generated by the rows of $(M^{-1})^{(\otimes t)}$ and an affine shift, with the property that the rate of $C$ is at least $\mathrm{Capacity}(\C)-\epsilon$, and $C$ can be encoded and decoded in time $\Oh(N \log N)$ and failure probability at most 
$\exp(-N^\beta)$.
\end{theorem}
\begin{proof}
	Follows by composing Lemma~\ref{lem:every-matrix-works}, Lemma~\ref{lem:matrix-implies-arikan}, Theorem~\ref{thm:local-to-global}, and \ref{thm:exp-code}.
\end{proof}
\jnote{Reviewer asks for expanding all the parameters.}

\begin{theorem}
\label{thm:thm2}
For every prime $q$, every symmetric memoryless channel $\C$ over $\F_q$, and every $\beta < 1$, there exists $k$, a  mixing matrix $M\in\F_q^{k\times k}$, and a polynomial $p$ such that for every $\epsilon > 0$ and every $N = k^t \geq p(1/\epsilon)$, 
there is an affine code $C$, that is generated by the rows of $(M^{-1})^{(\otimes t)}$ and an affine shift, with the property that the rate of $C$ is at least $\mathrm{Capacity}(\C)-\epsilon$, and $C$ can be encoded and decoded in time $\Oh(N \log N)$ and failure probability at most 
$\exp(-N^\beta)$.
\end{theorem}
\begin{proof}
	Follows by composing Corollary~\ref{cor:beta-close-to-one}, Lemma~\ref{lem:matrix-implies-arikan}, Theorem~\ref{thm:local-to-global}, and \ref{thm:exp-code}.
\end{proof}

\begin{theorem}
\label{thm:thm3}
Suppose $M\in\F_q^{k \times k}$ and $\beta > 0$ satisfy the condition that for every memoryless symmetric additive channel\footnote{An additive symmetric channel is a special case of symmetric channels, where the output is the sum of the input with an ``error'' generated independently of the input.} $\C$ and for every $\epsilon > 0$, for sufficiently large $n=k^s$, there is an affine code $C$ of length $n$ generated by the rows of $(M^{-1})^{(\otimes s)}$ of rate at least $\mathrm{Capacity}(\C)-\epsilon$ such that $C$ can be decoded with failure probability at most $\exp(-n^\beta)$. 

Then, for every $\beta' < \beta$ and every symmetric channel $\C'$, there is a polynomial $p$ such that for every $\epsilon > 0$ and every $N = k^t \geq p(1/\epsilon)$
there is an affine code $C$, that is generated by the rows of $(M^{-1})^{(\otimes t)}$ and an affine shift, with the property that the rate of $C$ is at least $\mathrm{Capacity}(\C')-\epsilon$, and $C$ can be encoded and decoded in time $\Oh(N \log N)$ and failure probability at most 
$\exp(-N^{\beta'})$.
\end{theorem}
We prove this theorem in Section~\ref{sec:lift}.

Note that in this theorem, we assume that $M$ achieves failure probabilities $\exp(-N^\beta)$ for \emph{additive} channels (which is only a subclass of all symmetric channels), to conclude that it achieves failure probability $\exp(-N^{\beta'})$ for \emph{all} symmetric channels. This is potentially useful, as proving good properties of polar codes for additive channels is often simpler --- in this setting there is a very clean equivalence between coding and linear compression schemes.
%\jnote{Emphasise that we lift additive channels $\to$ symmetric channels. The reviewer was confused (thinking that the result is actually weaker, and additivity applies to the conclusion of the theorem, not the assumptions)}

\section{Structural analysis of matrices}
\subsection{Exponential polarization for all mixing matrices \label{sec:all-matrix}}
We will first prove that a single specific matrix, namely $\HadMat$, after taking second Kronecker power satisfies exponential polarization. In \cite{BGNRS} local polarization of any mixing matrix was shown essentially by reducing to this case. Here we make this reduction more explicit, so that it commutes with taking Kronecker product of a matrix with itself. That is, we will later show that for any mixing matrix $M$ exponential polarization of $M^{\otimes 2}$  can be reduced to exponential polarization of $\HadMat^{\otimes 2}$.
\begin{lemma}
	Consider $M = \HadMat$ for nonzero $\alpha \in \F_q$. For every $\varepsilon > 0$ matrix $M^{\otimes 2}$ satisfies $(\frac{1}{4}, 2-\varepsilon)$ exponential polarization.
	\label{lem:4-by-4}
\end{lemma}
\begin{proof}
	Consider arbitrary sequence of i.i.d. random variables $(U_1, A_1), \ldots (U_4, A_4)$ with $H(U_i | A_i) = \delta$, as in the definition of exponential polarization. We can explicitly write down matrix $M^{\otimes 2}$ as
	\begin{equation*}
		M^{\otimes 2} = \left[\begin{array}{cccc} 1 & 0 & 0 & 0 \\ \alpha & 1 & 0 & 0 \\ \alpha & 0 & 1 & 0 \\ \alpha^2 & \alpha & \alpha & 1\end{array} \right].
	\end{equation*}
	Matrix $M^{\otimes 2}$ has four rows --- to achieve $\eta = \frac{1}{4}$ parameter of exponential polarization, we just need to show that there is at least one index $i$ satisfying the inequality as in the definition of exponential polarization (Definition~\ref{def:matrix-exp-polar}).
	Let us consider vector $U = (U_1, \ldots U_4)$ and similarly $A = (A_1, \ldots A_4)$. We want to bound
	\begin{align*}
		\H( (UM^{\otimes 2})_4 | (UM^{\otimes 2})_{<4}, A) & = \H(U_4 | U_1 + \alpha U_2 + \alpha U_3 + \alpha^2 U_4, U_2 + \alpha U_4, U_3 + \alpha U_4, A) \\
		& \leq \H(U_4 | U_2 + \alpha U_4, U_3 + \alpha U_4, A)
	\end{align*}

	By Lemma~\ref{lem:entropy-gives-prediction} there exist some function $f: \Sigma \to \F_q$, such that $\P(f(A_i) \not= U_i) \leq \delta$. Now, given vector $A$ and $W_2 := \alpha U_4 + U_2, W_3 := \alpha U_4 + U_3$, we can try to predict $U_4$ as follows: if $W_2 - f(A_2) = W_3 - f(A_3)$ we report $\hat{U}_4 := \alpha^{-1}(W_2 - f(A_2))$. Otherwise, we report $\hat{U}_4 := f(A_4)$.
	
	We want to show that $\P(\hat{U}_4 \not= U_4) \leq 3 \delta^2$. Indeed, $\hat{U}_4 \not= U_4$ only if at least two of the variables $U_i - f(A_i)$ for $i\in \{2,3,4\}$ are non-zero. By symmetry, we have $\P(\hat{U}_4 \not= U_4) \leq 3 \P(U_1 \not= f(A_1) \land U_2 \not= f(A_2)) = 3 \P(U_1 \not= f(A_1))^2 \leq 3 \delta^2$.

	By Fano's inequality \ref{lem:prediction-gives-entropy}, we have $\H(U_4 | U_2 + \alpha U_4, U_3 + \alpha U_4, A) \leq 6 \delta^2 (\log \delta^{-1} + \log q + \log 3)$. For any given $\varepsilon$, there exist $\tau$ such that if $\delta < \tau$ we have $6 (\log \delta^{-1} + \log q + \log 3) \leq \delta^{-\varepsilon}$, hence for those values of $\delta$ we have $\H( (UM^{\otimes 2})_{4} | (UM^{\otimes 2})_{<4}, A) \leq \delta^{2 - \varepsilon}$.
\end{proof}

We will now proceed to show that exponential polarization for $M^{\otimes 2}$ of any mixing matrix $M$ can be reduced to the theorem above. To this end we define the following containment relation for matrices.

\begin{definition}[Matrix (useful) containment]
		We say that a matrix $M \in \F_q^{k\times k}$ contains a matrix $R \in \F_q^{m\times m}$, if there exist some $T \in \F_q^{k\times m}$ and a permutation matrix $P \in \F_q^{k \times k}$, such that $PMT = \left[\begin{array}{c} R \\ 0 \end{array}\right]$. If moreover the last non-zero row of $T$ is rescaling of the standard basis vector $T_j = \alpha e_m$, we say that containment is $R$ in $M$ is useful and we denote it by $R \usef M$. 
	Note that useful containment is \emph{not} a partial order.
\end{definition}
\jnote{A comment on the definition above.}
The following fact about useful containment will be helpful.
\begin{claim}
		If $R \usef M$, then for any upper triangular matrix $U$ with diagonal elements $U_{i,i} = 1$, we also have $R \usef MU^{-1}$. 
		\label{claim:useful-upper}
\end{claim}
\begin{proof}
	Consider matrix $T$ and permutation $P$ as in the definition of useful containment for $R \usef M$. We can pick the very same permutation $P$ and matrix $T' = UT$ to witness $R \usef MU^{-1}$. All we have to show is that last non-zero row of $T'$ is standard basis vector $e_m$. Indeed, if $j_0$ is the last non-zero row of $T$, and $j > j_0$, rows $(U)_{j}$ are supported exclusively on elements with indices larger than $j_0$, hence $(UT)_{j} = (U)_j T = 0$. On the other hand $(U T)_{j_0} = \sum_{i} U_{j_0, i} T_i = \sum_{i \geq j_0} U_{j_0, i} T_i = \alpha e_m$, where the last equality follows from the fact that $T$ was useful --- that is $T_{j_0} = \alpha e_m$ and $T_i = 0$ for $i > j_0$.

\end{proof}

Results of the Lemma 5.5 in \cite{BGNRS} can be reintepreted as the following Lemma. We give a full new proof here, as we describe it now in the language of useful containment.
\begin{lemma}
	Every mixing matrix $M \in \F_q^{k\times k}$ contains matrix $H = \HadMat$ in a useful way.
\end{lemma}
\begin{proof}
		For any matrix $M$, there is some permutation matrix $P$ and pair $L, U$, such that $PM = L U$ where $L$ is lower triangular, and $U$ is upper triangular. Matrix $M$ being mixing is equivalent to the statement that $L$ and $U$ are invertible, and moreover $L$ is not diagonal. As such by Claim~\ref{claim:useful-upper} it is enough to show that any lower-triangular $L$, which is not diagonal, contains $H$ in a useful way. Indeed, let $s$ be the last column of $L$ that contains more than a single non-zero entry, and let $r$ to be the last row of non-zero entry in column $L_{\cdot, s}$. Note that column $L_{\cdot, r}$ has single non-zero entry $L_{r,r} = 1$. We will show a matrix $T \in \F_q^{k \times 2}$ as in the definition of useful containment. Let us specify a second column of $T_{\cdot, 2} := e_r$. To specify the first column of $T$ we wish to find a linear combination of columns of $L_{1, \cdot}, \ldots, L_{r-1, \cdot}$ such that $\sum_{i\leq r-1} t_{i} L_{i, \cdot} = \alpha e_s + \alpha e_r$. Then coefficeints $t_i$ can be used as the first column of matrix $M$. We can set those coefficients to $t_{i} = - L_{s, i}$ for $i \in [s+1, r-1]$, and $t_s = 1$ --- this setting is correct, because columns $L_{i,\cdot}$ for $i \in [s+1, r-1]$ has only one non-zero entry $L_{i,i}$. Now if $P$ is any matrix corresponding to a permutation which maps $s \mapsto 1$ and $r \mapsto 2$, the containemnt $H \usef L$ is witnessed by pair $P$ and $T$.
\end{proof}
\begin{lemma}
	If matrix $R \usef M$ where $R\in \F_q^{s\times s}$ and $M \in \F_q^{k \times k}$, then $R^{\otimes 2} \usef M^{\otimes 2}$.
\end{lemma}
\begin{proof}
	Consider matrix $T$ and permutation $P$ as in the definition of useful containment for $R\usef M$. Note that $P^{\otimes 2}M^{\otimes 2} T^{\otimes 2} = (PMT)^{\otimes 2}$. 
	As such, restriction of a matrix $P^{\otimes 2} M^{\otimes 2} T^{\otimes 2}$ to rows corresponding to $[k]\times [k]$ is exactly $R$, and all remaning rows are zero. We can apply additional permutation matrix $\tilde{P}$ so that those are exactly first $k^2$ rows of the matrix $\tilde{P} P^{\otimes 2} M^{\otimes 2} T^{\otimes 2}$ give matrix $R^{\otimes 2}$, and the remaining rows are zero.
\end{proof}
% \jnote{We would like statement if $M$ contains $R$ usefully, and $R$ satisfied exp-suction, then so does $M$. This is not true as things are defined now.}
\begin{lemma}
	If matrix $M$ contains matrix $R = \HadMat^{\otimes 2}$ in a useful way, then matrix $M$ satisfies $(\frac{1}{k}, 2 - \varepsilon)$ exponential polarization.
\end{lemma}
\begin{proof}
Take $P \in \F_q^{k\times k}$ and  $T \in \F_q^{k\times 4}$ as in the definition of containment. Let moreover $j$ be the last non-zero row of $T$. We have 
\begin{align*}
\H( (UM)_{j} | (UM)_{<j}, A) & = \H( (UM)_{j} T_{j, 4} + (UM)_{<j} T_{<j, 4} | (UM)_{<j}, A) \\
& = \H( (UMT)_4 | (UM)_{<j}, A) \\
& \leq \H( (UMT)_4 | (UM)_{<j} T_{<j, <4}, A).
\end{align*}
Observe now that $(UM)_{<j} T_{<j, <4} = (UMT)_{<4}$. Indeed --- according to the definition of useful containment and because $j$ is last non-zero row of $T$, we have $T_{j, <4} = 0$ ($j$-th row has only one non-zero entry $T_{j, 4}$, as well as $T_{>j, <4} = 0$. Therefore
\begin{align*}
		\H( (UM)_{j} | (UM)_{<j}, A) & \leq \H( (UMT)_4 | (UMT)_{<4}, A) \\
& = \H( (U P^{-1} R)_4 | (UP^{-1} R)_{<4}, A) \\
& = \H( (UR)_4 | (UR)_{< 4}, A),
\end{align*}
where the last equality follows from the fact that $U$ and $UP^{-1}$ are identically distributed (i.e. entries in $U$ are i.i.d).

This conditional entropy was bounded in the proof of Lemma~\ref{lem:4-by-4}.
\end{proof}

\subsection{Maximally polarizing matrix}
In this subsection we will prove Lemma~\ref{lem:code-exp-suction}.
\begin{proof}[Proof of Lemma~\ref{lem:code-exp-suction}]
	Let us again consider a sequence of i.i.d. pairs $(U_i, A_i)$ for $i \in [k]$, such that $H(U_i | A_i) = \delta$. By Lemma~\ref{lem:entropy-gives-prediction}, there is some $f : \Sigma \to \F_q$ such that $\P(f(A_i) \not= U_i) \leq \delta$. Let us take $\tilde{U}_i := U_i - f(A_i)$.

	We wish to bound $\H( (UM)_j | (UM)_{<j}, A)$, for all $j > (1 - \eta) k$. We have 
	\begin{equation*}
			\H( (UM)_j | (UM)_{<j}, A) \leq \H(U | UM_0, A) = H(\tilde{U} | \tilde{U}M_0, A) \leq \H(\tilde{U} | \tilde{U} M_0),
	\end{equation*}
	where the inequalities follow from the fact that for random variables $(X, Y, S, T)$ it is always the case that $\H( X | S, T) \leq \H(X, Y | S, T) \leq \H(X, Y | S)$.

	Given $\tilde{U} M_0$ we can produce estimate $\hat{U} := \argmin_{V} \{ \wt(V) : V M_0 = \tilde{U} M_0\}$, where $\wt(V) = |\{j : V_j \not=0 \}|$.

	Let us observe that if $\wt(\tilde{U}) \leq b$ then $\hat{U} = \tilde{U}$. Indeed, we have $\wt(\hat{U}) \leq \wt(\tilde{U})$, therefore $\wt(\hat{U} - \tilde{U}) \leq 2 \wt(\tilde{U}) \leq 2 b$, but on the other hand $(\hat{U} - \tilde{U}) M_0 = 0$, and by the assumption on $\ker M_0^T$ we deduce that $\hat{U} - \tilde{U} = 0$. Therefore $\P(\tilde{U} \not= \hat{U}) \leq \P(\wt(\tilde{U}) > b)$. All coordinates of $\tilde{U}$ are independent, and each $\tilde{U}_i$ is nonzero with probability at most $\delta$, therefore
	\begin{equation*}
		\P(\wt(\tilde{U}) > \beta_1) \leq \binom{k}{b} \delta^{b}
	\end{equation*}
	and by Fano inequality (Lemma \ref{lem:prediction-gives-entropy}), we have
	\begin{equation*}
		H(\tilde{U} | \tilde{U} M_0) \leq 2 C \delta^{b} (b \log \delta^{-1} + b \log C + \log q )
	\end{equation*}
	where $C = \binom{k}{b}$. Again, for any $\varepsilon$, and small enough $\delta$ (with respect to $\varepsilon, b, C, q$), we have $H(\tilde{U} | \tilde{U} M_0) \leq \delta^{b - \varepsilon}$.

	This shows that for any $j > (1 - \eta) k$ and small enough $\delta$ we have
	\begin{equation*}
		\H( (UM)_j | (UM)_{<j}, A) \leq \delta^{b - \varepsilon},
	\end{equation*}
	which completes the proof of a exponential polarization for matrix $M$.
\end{proof}

\subsection{Source coding implies good distance}

\begin{proof}[Proof of Lemma~\ref{lem:compression-implies-distance}.]
	Consider maximum likelihood decoder $\Dec'(y) := \argmax_{x \in \F_q^k} \P(U = x | UP = y)$. By definition, we have $\P(\Dec'(UP) \not= U) < \P(\Dec(UP) \not= U) < \exp(-k^{\gamma})$.

	Note that for $U$ distributed according to $B_{q}(\varepsilon)$, we have $\Dec'(y) = \argmin_{x : xP = y} \wt(x)$, where $\wt(x)$ is number of non-zero elements of $x$.

	Consider set $E = \{x \in \F_q^k : \exists h\in \ker M, \wt(x+h) < \wt(x)\}$, and observe that $\P(\Dec'(UP) \not = U) \geq \P(U \in E)$. We say that vector $u \in \F_q^k$ is \emph{dominated} by $v \in \F_q^k$ (denoted by $u \dom v$) if and only if $\forall i \in \supp(u),\, u_i = v_i$. We wish to argue that for any $w_1 \in E$ and any $w_2 \succeq w_1$, we have $w_2 \in E$. Indeed, if $w_1 \in E$, then there is some $h \in \ker M$ such that $\wt(w_1 + h) < \wt(w_1)$. We will show that $\wt(w_2 + h) < \wt(w_2)$, which implies that $w_2 \in E$. Given that $w_1 \preceq w_2$, we can equivalently say that there is a vector $d$ with $w_1 + d = w_2$ and $\wt(w_2) = \wt(w_1) + \wt(d)$. Hence
	\begin{equation*}
		\wt(w_2 + h) = \wt(w_1 + d + h) \leq \wt(w_1 + h) + \wt(d) < \wt(w_1) + \wt(d) = \wt(w_1 + d) = \wt(w_2)
	\end{equation*}

	Consider now $w_0 \in \ker P$ to be minimum weight non-zero vector, and let us denote $A = \wt(w_0)$. We wish to show a lower bound for $A$. By definition of the set $E$ we have $w_0 \in E$, and by upward closure of $E$ with respect to domination we have $\P(U \in E) \geq \P(w_0 \preceq U) = (\frac{\varepsilon}{q - 1})^A$.

	On the other hand we have $\P(U \in E) \leq \P(\Dec'(UP) \not= U) \leq \P(\Dec(UP) \not= U) \leq \exp(-k^{\gamma})$. By comparing these two inequalities we get
	\begin{equation*}
		A \geq \frac{k^\gamma}{\log (q/\varepsilon)} \ . \qedhere
	\end{equation*}
\end{proof}

\section{Strong polarization from limiting exponential polarization, generically \label{sec:lift}}

%Exponential polarization implies strong exponential polarization}
 
Suppose we know that polar codes associated with a matrix $M \in \F_q^{k \times k}$ achieve capacity with error probability $\exp(-N^{\beta})$ in the limit of block lengths $N \to \infty$. In this section, we prove a general result that `lifts" (in a black box manner) such a statement to the claim that, for any $\beta' < \beta$, polar codes associated with $M$ achieve polynomially fast convergence to capacity (i.e., the block length $N$ can be as small as $\mathrm{poly}(1/\epsilon)$ for rates within $\epsilon$ of capacity), and $\exp(-N^{\beta'})$ decoding error probability \emph{simuletaneously}. Thus convergence to capacity at finite block length comes with almost no price in the failure probability. Put differently, the result states that one can get polynomial convergence to capacity for free once one has a proof of convergence to capacity in the limit with good decoding error probability. This latter fact was shown in \cite{KSU10} for the binary alphabet and \cite{mori-tanaka} for general alphabets. 

 \jnote{On the other hand, if I remember correctly for limiting polarization, there is an exact characterization what is the correct exponent $\beta$, depending on a matrix $M$. We are lifting this theorem to the finite block-length scenario, might be worth mentioning somewhere. Also --- I think that the original motivation for considering different kernels, was that one could hope to get $\beta$ arbitrairly close to $1$. We are doing this in Corollary~\ref{cor:beta-close-to-one} and in this section.}
 \vnote{Excellent points - I think Madhu has even mentioned this in abstract and hopefully will reiterate it in intro too}

 \vnote{Used the notion of a matrix, rather than the associated Arikan martingale, exponentially/strongly polarizing. Need to define this abuse of notation earlier. Also part about why we have the direct approach can be tightened and/or moved to the introduction. Must also define mixing matrix.}
 
 %We now state and prove our result formally. In fact, we only need the polarizing property of $M^{\otimes t}$ for compression, that too of the simple symmetric ``Bernoulli'' sources which give equal probability to all non-zero elements. Recall the source $B_q(\gamma)$ (Definition~\ref{defn:q-ary-bernoulli}) which gives mass $1-\gamma$ to $0$ and $\gamma)/(q-1)$ to non-zero elements of $\F_q$. The entropy (measured as number of $q$-ary symbols) of this source is $h_q(\gamma)$, so an optimal compression algorithm must compress $N$ samples from $B_q(\gamma)$ to a string over $\F_q$ of length $\approx h_q(\gamma) N$.
  
%\begin{theorem}
%	\label{thm:lifting-polar-codes}
%Let $M \in \F_q^{k \times k}$ be a mixing matrix, and let $\gamma  \in (0,1)$. Suppose that for all $\epsilon > 0$, the following holds for all large enough $t$: there is some subset $S$ of $(h_q(\gamma)+\epsilon) k^t$ columns of $M^{\otimes t}$ that defines a linear compression scheme (for $k^t$ i.i.d copies of $B_q(\gamma)$), along with an accompanying decompression scheme with error probability (over the randomness of the source) at most $\exp(-k^{\beta t})$. Then for all $\beta' < \beta$, there exists $t_0 = t_0(\beta',\beta)$ such that the Arikan martingale associated with some column permuted version of $M^{\otimes t_0}$, is $\beta' t_0 \log_2 k$-exponentially strongly polarizing.
%\end{theorem}

 \begin{proof}[Proof of \ref{thm:thm3}]
	 Consider the channel that outputs $X + Z$ on input $X$, where $Z \sim B_q(\gamma)$ for some $\gamma > 0$ (depending on $\beta, \beta'$). 
	 The hypothesis on $M$ implies that for sufficiently large $N$ the polar code corresponding to $M$ will have failure probability at most $\exp(-N^{\beta})$ on this channel. Using the well-known equivalence between correcting errors for this additive channel, and linear compression schemes, we obtain that for all large enough $t$ there is some subset $S$ of $(h_q(\gamma)+\epsilon) k^t$ columns of $M^{\otimes t}$ that defines a linear compression scheme (for $k^t$ i.i.d copies of $B_q(\gamma)$), along with an accompanying decompression scheme with error probability (over the randomness of the source) at most $\exp(-k^{\beta t})$. 
	 
	 We now claim that  for all $\beta' < \beta$, there exists $t_0 = t_0(\beta',\beta)$ such that the Arikan martingale associated with some column permuted version of $M^{\otimes t_0}$, is $\beta' t_0 \log_2 k$-exponentially strongly polarizing.

The proof of this claim  is in fact immediate, given the ingredients developed in previous sections. Apply the hypothesis about $M$ in the theorem with the choice $\epsilon  = (\beta-\beta')/4$ and $\gamma$ chosen small enough as a function $\beta,\beta'$ so that $h_q(\gamma)  \le (\beta-\beta')/4$ and let $t_0$ be a large enough promised value of $t$. Put $m = k^{t_0}$, and $\ell =  (h_q(\gamma)+\epsilon) m$ and $L = M^{\otimes t_0}$. 
Using Lemma~\ref{lem:compression-implies-distance}, we know there is submatrix $L' \in \F_q^{m \times \ell}$ of $L$ such that $\mathrm{ker}((L')^T)$ defines a code of distance $\Delta \ge m^{\beta}/\log^{-1}(q/\gamma)$. Define $M_0 = [ L' \mid \cdot] \in \F_q^{m \times m}$ to be any matrix obtained by permuting the columns of $L$ such that the columns in $L'$ occur first. By Lemma~\ref{lem:code-exp-suction},  the matrix $M_0$ is $(1 -\ell/m,\Delta)$-polarizing.
For our choice of $\gamma,\epsilon$, $\ell/m \le \frac{\beta -\beta'}{2}$ and $\Delta \ge m^{(\beta+\beta')/2}$. using Lemma~\ref{lem:matrix-implies-arikan} and Theorem~\ref{thm:local-to-global}, it follows that the Arikan martingale associated with $M_0$ exhibits 
$(\beta+\beta')/2 \times \left(1 - \frac{\beta-\beta'}{2} \right) \log_2 m$-exponentially strong polaraization. 
Since $(\beta+\beta')/2 \times \left(1 - \frac{\beta-\beta'}{2} \right) \ge \beta'$, the claim follows.

Applying Theorem~\ref{thm:exp-code} to the matrix $M_0 = M^{\otimes t_0}$ we conclude that there is a polynomial $p$ such that  given the gap to capacity $\varepsilon > 0$, and for every $s$ satisfying $N = k^{t_0 s} \geq \poly(\frac{1}{\varepsilon})$ there is an affine code generated by a subset of rows of $(M_0^{-1})^{\otimes s}$ which achieves $\varepsilon$-gap to capacity and has failure probability $\exp(-N^{\beta'})$. But this resulting code is simply an affine code generated by a subset of the rows of $(M^{-1})^{\otimes t}$, for $t = s t_0$, This concludes the proof. 
\end{proof}
\iffalse
The tensor power $M_0^{\otimes s}$ of the matrix $M_0$ from the above proof consists of some permutation of columns of $M^{\otimes (t_0s)}$. It thus follows that for any channel, there is a choice of rows of $(M^{\otimes t})^{-1}$ that defines a polar code of block length $N=k^t$ that achieves rate within additive gap $\rho^t$ of the channel capacity (for some fixed $\rho < 1$) and has decoding error probability falling as $\exp(-N^{\beta'})$. The results of \cite{KSU10,mori-tanaka} characterized, for any mixing $M$, a positive exponent $\beta=\beta(M) > 0$, for which polar codes based on $M$ achieve error probability at most 
$\exp(-N^{\beta})$. Assuming this even only for source coding of Bernoulli sources, we conclude the following, namely that we can approach the same optimal exponent $\beta(M)$ with polynomial convergence to capacity.

\begin{corollary}
Let $M \in \F_q^{k \times k}$ be a mixing matrix for a prime $q$ let $\beta' < \beta(M)$. Then for any  additive channel $W$ with inputs from $\F_q$, there are polar codes defined by $M$ (i.e., generated by some subset of rows of $(M^{\otimes t})^{-1}$) of rate $I(W)-\eps$ and block length $N \le \mathrm{poly}(1/\eps)$ that achieve decoding error probability $\exp(-N^{\beta'})$. 
\end{corollary}

\fi

%\vnote{Should we say that while we state the generic conversion of limiting to strong polarization in the exponentially polarizing regime, a similar statement should also hold for other error probabilities}

%\mnote{What else needs to be said?}

\bibliographystyle{plain}
\bibliography{biblio,polar-refs}

\newpage
\appendix
\section{Local to global exponential polarization \label{sec:local-to-global}}
The proof of Theorem~\ref{thm:local-to-global} is essentially the same as the proof of corresponding Theorem 1.6 in \cite{BGNRS}. Lemma~\ref{lem:single-piece-suction} and Lemma~\ref{lem:second-phase} are new in this paper, yet the proof of Lemma~\ref{lem:single-piece-suction} is similar to the proof of Lemma 3.3 there. Theorem~\ref{thm:local-to-global} is essentially repeating the argument from Theorem 1.6 in \cite{BGNRS}, except for using Lemma~\ref{lem:second-phase} in place of the lemma present therein, and hence arriving at stronger conclusion.

We remind a definition of adapted sequence from \cite{BGNRS}.
\begin{definition}
    \label{def:adapted-sequence}
    We say that a sequence $Y_1, Y_2 \ldots$ of random variables is \emph{adapted} to the sequence $X_1, X_2 \ldots$ if and only if for every $t$,  $Y_t$ is completely determined given $X_1, \ldots X_t$. We will use $\E[Z | X_{[1:t]}]$ as a shorthand for $\E[Z | X_1, \ldots X_t]$, and $\P[ E | X_{[1:t]}]$ as a shorthand for $\E[\1_E | X_1, \ldots X_t]$. If the  underlying sequence $X$ is clear from context, we will skip it and write just $\E[Z | \cF_t]$.
\end{definition}
\begin{lemma}
		\label{lem:single-piece-suction}
	There exist $C < \infty$ such that for all $\eta, b, \varepsilon$ following holds. Let $X_t$ be a martingale satisfying $\Pr(X_{t+1} < X_t^{b} | X_{t}) \geq \eta$, where $X_0 \in (0, 1)$. Then
	\begin{equation*}
			\P(\log X_T > (\log X_0 + CT) b^{(1-\varepsilon)\eta T}) < \exp(-\Omega( \varepsilon \eta T))
	\end{equation*}
\end{lemma}
\begin{proof}
	Let us consider random variables $Y_t := \log (X_t/X_{t-1})$. This sequence of random variables is adapted to the sequence $X_t$ in the sense of Definition~\ref{def:adapted-sequence}. Let us decompose $Y_t = Y_t^+ + Y_t^-$, where $Y_t^+ = Y_t \mathbf{1}_{Y_t \geq 0}$.
	Note that by Markov inequality 
	\begin{equation*}
		\P(Y_{t+1} > \lambda | X_{[1:t]}) = \P(X_{t+1} > X_t \exp(\lambda) | X_{[1:t]}) \leq \exp(-\lambda)\frac{\E[X_{t+1} | X_{[1:t]}]}{X_t} = \exp(-\lambda)
	\end{equation*}

	By Lemma~\ref{lem:sum-of-subexp} we deduce that for some $C$, we have
	\begin{equation*}
		\P(\sum_{i \leq T} Y_i^{+} > C T) \leq \exp(-\Omega(T))
	\end{equation*}
	On the other hand, if we take $Z_t$ to be the indicator variable for an event $X_t < X_{t-1}^{\beta_1}$. By Lemma~\ref{lem:sum-of-events} we have
	\begin{equation*}
		\P(\sum_{i \leq T} Z_i \leq (1-\varepsilon) \eta T) \leq \exp(-\Omega(T\varepsilon\eta))
	\end{equation*}
	If both of those unlikely events do not hold, that is we have simultaneously $\sum_{i \leq T} Y_i^{+} < CT$ and $\sum_{i \leq T} Z_i > (1-\varepsilon) \eta T$, we can deduce that $\log X_T \leq (\log X_0 + CT) b^{(1 - \varepsilon) \eta T}$ --- i.e. the largest possible value of $X_T$ is obtained if all the initial $Y_i$ were positive and added up to $CT$ (at which point value of the martingale would satisfy $\log X_{T'} \leq \log X_0 + CT$), followed by $(1-\varepsilon) \eta T$ steps indicated by variables $Z_i$ --- for each of those steps, $\log X_{t+1} \leq b \log X_t$.
\end{proof}

\begin{lemma}
	\label{lem:second-phase}
	For all $\eta, b, \varepsilon, \gamma$ the following holds. Let $X_t$ be a martingale satisfying $\Pr(X_{t+1} < X_t^{b} | X_t) \geq \eta$, where $X_0 < \exp(- \gamma T)$ with some $\gamma > 0$, then
	\begin{equation*}
\P(\log X_T < -b^{(1 - \varepsilon) \eta T}) < \exp(-\Omega_{\nu,\varepsilon,\eta,\gamma}(T)))
	\end{equation*}
\end{lemma}
\begin{proof}
		Consider sequence $t_0, t_1, \ldots t_m \in [T]$, where $t_0 = 0, t_m = T$, and $\frac{\gamma T}{C} \leq |t_i - t_{i-1}| \leq \frac{\gamma T }{2C}$, and therefore $m = \Oh(C \gamma^{-1})$, where $C$ is a constant appearing in the statement of Lemma~\ref{lem:single-piece-suction}. For each index $s \in [m]$ we should consider a martingale $X^{(s)}_i := X_{t_s + i}$, and we wish to apply Lemma~\ref{lem:single-piece-suction} to this martingale $\hat{X}^{(s)}$, with $T = t_{s+1} - t_s$. We can union bound total failure probability by $m \exp(-\Omega(\gamma \varepsilon \eta T))$. 

	In case we succeed, we can deduce that for each $i$ we have
	\begin{equation}
			\label{eq:lucky-event}
			\log X_{t_i} < (\log X_{t_{i-1}} + C(t_i - t_{i-1})) b^{(1-\varepsilon)\eta (t_{i} - t_{i-1})}. 
	\end{equation}
	We will show that by our choice of parameters, we can bound $C(t_i - t_{i-1}) \leq -\frac{1}{2} \log X_{t_i}$. Let us first discuss how this is enough to complete the proof. Indeed, in such a case we have
	\begin{equation}
		\log X_{t_i} < \frac{1}{2} (\log X_{t_{i-1}}) b^{(1 - \varepsilon) \eta (t_{i} - t_{i-1})}, 
		\label{eq:inductive-step}
	\end{equation}
	and by induction
	\begin{equation*}
	\log X_{t_m} < \frac{1}{2^m} (\log X_{0}) b^{(1-\varepsilon) \eta t_m}.
	\end{equation*}
	For fixed $\eta, m$ and $T$ large enough (depending on $\eta, m, \varepsilon$), this yields $\log X_T < - b^{(1 - 2\varepsilon) \eta T}$, and the result follows up by changing $\varepsilon$ by a factor of $2$.

	All we need to do is to show is that for every $i$ we have 
	\begin{equation}
			C(t_{i+1} - t_i) \leq -\frac{1}{2} \log X_{t_i},
			\label{eq:inductive-condition}
	\end{equation}
	assuming that inequalities~(\ref{eq:lucky-event}) hold for every $i$.
	We will show this inductively, together with $\log X_{t_i} \leq - \gamma T$. Note that we assumed this inequality to be true for $X_{t_0} = X_0$. By our choice of parameters we have $C(t_{i+1} - t_i) \leq \frac{\gamma T}{2}$, therefore for $t_{i+1}$ the inequality~(\ref{eq:inductive-condition}) is satisfied. 
	
	We will now show that $\log X_{t_{i+1}} \leq \log X_{t_i} \leq - \gamma T$ to finish the proof by induction. We can apply inequality~(\ref{eq:inductive-step}) to $X_{t_i}$, to deduce that $\log X_{t_{i+1}} \leq \frac{1}{2} (\log X_{t_i}) b^{\frac{1}{2} \frac{\gamma}{C} T}$. This for large values of $T$ (given parameters $b, \gamma$ and $C$) yields $\log X_{t_{i+1}} < \log X_{t_i}$ --- indeed this inequality will be true as soon as $b^{\frac{\gamma}{2C} T} > 2$, because both $\log X_{t_{i+1}}$ and $\log X_{t_i}$ are negative, which completes the proof.
\end{proof}

Before we proceed with the proof, let us recall the following lemma from~\cite{BGNRS}, stating that locally polarizing martingales are exponentially close to boundary $\{0, 1\}$ for some basis $(1-\nu)$, except with exponentially small failure probability.
\begin{lemma}[Lemma 3.1 from \cite{BGNRS}]
    \label{lem:potential-function}
	If a $[0,1]$-martingale sequence $X_0, \ldots X_t, \ldots,$ is $(\alpha,\tau(\cdot),\theta(\cdot))$-locally polarizing, then there exist $\nu > 0$, depending only on $\alpha, \tau, \theta$, such that 
\[\E [\min(\sqrt{X_t}, \sqrt{1-X_t}) ] \leq (1 - \nu)^{t}.\]
\end{lemma}
We will also need Lemma 3.3 from \cite{BGNRS} --- it plays the same role as Lemma~\ref{lem:second-phase} to control strong polarization of the martingale at the high end (where the exponential suction condition does not apply).
\begin{lemma}[Lemma 3.3 from \cite{BGNRS}]
	There exists $c < \infty$, such that for all $K, \alpha$ with $K \alpha \geq c$ the following holds.
    Let $X_t$ be a martingale satisfying $\P\left(X_{t+1} < e^{-K} X_t | X_t\right) \geq \alpha$, where $X_0 \in (0, 1)$. Then $\P(X_T > \exp(- \alpha K T/4)) \leq \exp(-\Omega(\alpha T))$.
    \label{lem:strong-polarization}
\end{lemma}

We are now ready to prove local to global lifting theorem for exponential polarization.
\begin{proof}[Proof of Theorem~\ref{thm:local-to-global}]
	Consider locally polarizing martingale, and let us fix some $\varepsilon > 0$. By Markov inequality applied to~\ref{lem:potential-function} with $t = \varepsilon T$ we deduce that for some $\nu$ we have
	\begin{equation*}
			\P(\max(X_{\varepsilon T}, 1 - X_{\varepsilon T}) \geq (1 - \frac{\nu}{4})^{\varepsilon T}) < \exp(-\Omega_\varepsilon(T))
	\end{equation*}

	Consider $\tau_0$ to be such that if $X_t < \tau_0$, we have probability at most $\eta$ that $X_{t+1} < X_t^b$ (existence of such a value is guaranteed by exponential local polarization), and moreover if $1 - X_t < \tau_0$, we have probability at least $\alpha$ for $(1 - X_{t+1}) < \exp(-K) (1 - X_{t+1})$, where $K$ is large constant depending on $\alpha$ and the target rate of polarization --- this is guarantee by suction at the high end condition in local polarization definition of a martingale $X_t$.

	Let us condition on $\max(X_{\varepsilon T}, 1 - X_{\varepsilon T}) < (1 - \frac{\nu}{4})^{\varepsilon T}$. By the Doobs martingale inequality (Lemma~\ref{lem:doobs}), we can deduce that $\P(\max_{t \in [\varepsilon T, T]} \max(X_{t}, 1 - X_t) > \tau) \leq \tau^{-1} (1 - \frac{\nu}{4})^{-\varepsilon T} \leq \exp(-\Omega_{\tau,\nu,\varepsilon}(T))$. Let us now condition in turn on this event not happening.

	We will consider first the case when $X_{\varepsilon T} < (1 - \frac{\nu}{4})^{\varepsilon T}$, and let us put $\gamma := - \varepsilon \log (1 - \frac{\nu}{4})$, so that $X_{\varepsilon T} < \exp(- \gamma T)$.

	We can now apply Lemma~\ref{lem:second-phase} to the martingale sequence starting with $X_{\varepsilon T}$ --- the assumption of those lemmas are satisfied, as long as $X_t$ stays bounded by $\tau$ (by the exponential local polarization property), hence we deduce that in this case, except with probability $\exp(-\Omega_{\gamma, \varepsilon,\eta}(T))$, we have
	\begin{equation*}
		\log X_{T} < -b^{(1-\varepsilon)^2 \eta T},
	\end{equation*}
	and therefore $X_T < 2^{-b^{(1-\varepsilon)^2 \eta T}}$.

	On the other hand, if $1 - X_{t} < \tau$ for all $\varepsilon T \leq t \leq T$, the suction at the high end condition of local polarization applies, and we can apply Lemma~\ref{lem:strong-polarization} to martingale $1 - X_{\varepsilon T + t}$ to deduce that except with probability $\exp(-\Omega_{\alpha}(T))$, we have $1 - X_T < \exp(-\alpha K (1-\varepsilon T) / 4) < \gamma^T$ for suitable choice of $K$ depending on $\gamma$ and $\alpha$.
\end{proof}

\section{Arikan Martingale \label{sec:arikan}}

In this section, we provide a definition of Arikan Martingale.

For every matrix invertible matrix $M$ and channel $\mathcal{C} : \F_q \to \mathcal{Y}$, we define a martingale sequence $X_t$, for $t = 0, 1, \ldots$, where all $X_t \in [0,1]$.

Intuitively, for a given matrix $M$ and $t \in \mathbb{N}$, the marginal distribution of $X_t$ is the same as distribution of $\H( (\bvec{Z} M^{\otimes t})_j\, |\, (\bvec{Z} M^{\otimes t})_{<j},  \bvec{Y} )$ over a random index $j \in [k^t]$, where $\bvec{Z}_i \sim \mathrm{Unif}(\F_q)$ are independent, and $\mathcal{Y}_i$ sampled independently according to $Y_i \sim C_{Y | Z = \bvec{Z}_i}$. That is, we apply matrix $M^{\otimes t}$ to a vector with independent coordinates $Z_i$, and we look at the entropy of the random output coordinate, conditioned on all previous ones. The entries $Z_i$, conditioned on $Y_i$ have normalized entropy equal to $1-\mathrm{Capacity}(C)$ for symmetric channel $C$, in particular $X_0 = 1 - \mathrm{Capacity}(C)$. If the variable $X_t$ is strongly polarized, it means that about $1-\mathrm{Capacity}(C)$ fraction of all $(\bvec{Z} M^{\otimes t})_{j}$ have entropy close to one (after conditioning on all the previous entries), and most of remaining variables has entropy close to zero --- they can be predicted from the previous values with huge probability.

The martingale structure of $X_t$ with respect to $t$ is a consequence of chain rule for entropy together with recursive decomposition of multiplication by matrix $M^{\otimes t}$. The relation between our definition of exponential matrix polarization (Definition~\ref{def:matrix-exp-polar}) and the local behavior of the \Arikan~martingale is consequence of the fact that $\bvec{A'}$ (in the definition below) is obtained from independent copies $\bvec{A}$ via multiplication by $\bvec{M}$. The notational difficulty in proving this equivalence (Lemma~\ref{lem:matrix-implies-arikan}) follows from the fact that conditioning in the conditional entropies under consideration is syntatically different --- although equivalent.

In what follows, the vectors in $\F_q^{k^t}$ are indexed by tuples $\bvec{j} \in [k]^t$, $\preceq$ denotes a lexicographic order on tuples. For $\bvec{A} \in \F_q^{k^t}$ and $\bvec{j} \in [k]^t$, we use notation $\bvec{A}_{\preceq \bvec{j}}$ to denote all entries of $A$ with indices preceeding $j$ according to lexicographic order $\preceq$. Moreover for a tuple of indices $\bvec{j} \in [k]^{t-1}$, and a vector $\bvec{A} \in \F_q^{k^t}$, we use notation $A_{[\bvec{j}, \cdot]} \in \F_q^k$ to denote a vector $(A_{[\bvec{j}, 1]}, \ldots A_{[\bvec{j}, k]})$.

\begin{definition}[\Arikan\ martingale, Defintion 4.1 in \cite{BGNRS}]
    \label{def:arikan-martingale}
    Given an invertible matrix $M \in \F_q^{k\times k}$ and a channel description $C_{Y|Z}$ for $Z \in \F_q, Y \in \mathcal{Y}$, the \Arikan-martingale $X_0, \ldots X_t, \ldots$ associated with it is defined as follows.
For every $t \in \mathbb{N}$,
let $D_t$ be the distribution on pairs $\F_q^{k^t} \times \mathcal{Y}^{k^t}$ described inductively below:

A sample $(A,B)$ from $D_0$ supported on $\F_q \times \mathcal{Y}$ is obtained by sampling $A \sim \F_q$, and $B \sim C_{Y|Z=A}$.
For $t \geq 1$, a sample $(\bvec{A}', \bvec{B}') \sim D_{t}$ supported on $\F_q^{k^t}\times \mathcal{Y}^{k^t}$ is obtained as follows:
\begin{itemize}
    \item Draw $k$ independent samples $(\bvec{A}^{(1)}, \bvec{B}^{(1)}), \dots, (\bvec{A}^{(k)}, \bvec{B}^{(k)}) \sim D_{t-1}$.
    \item Let $\bvec{A}'$ be given by
$\bvec{A}'_{[\bvec{i}, \cdot]}
= (\bvec{A}^{(1)}_{\bvec{i} } ~ ,\dots, ~ \bvec{A}^{(k)}_{\bvec{i}})\cdot  M$
for all $\bvec{i}\in [k]^{t-1}$
and $\bvec{B}' = (\bvec{B}^{(1)}, \bvec{B}^{(2)}, \ldots \bvec{B}^{(k)})$.
\end{itemize}

%$Y_t := M^{\otimes t} \bvec X^{(t)}$.

Then, the sequence $X_t$ is defined as follows:
For each $t \in \mathbb{N}$, sample $i_t \in [k]$ iid uniformly.
Let $\bvec j = (i_1,\ldots,i_t)$
and let
$X_t := \H(\bvec{A}_{\bvec j} | \bvec{A}_{\prec \bvec j}, \bvec{B})$,
where the entropies are with respect to the distribution
$(\bvec{A}, \bvec{B}) \sim D_t$. The only randomness in the process $X_t$ comes from the selection of random multi-index $\bvec{j}$.
\end{definition}

Before we proceed with the proof of Lemma~\ref{lem:matrix-implies-arikan} relating exponential polarization of the matrix and exponential polarization of the associated \Arikan~martingale, let us remind the following lemma from \cite{BGNRS}.

\begin{lemma}
Let $\bvec{A}^{(1)}, \ldots \bvec{A}^{(k)}$, and $\bvec{A}'$ be defined as in Definition~\ref{def:arikan-martingale}, and let $V, W$ be arbitrary random variables. Then for any multiindex $\bvec{i} \in [k]^{t}$ and any $i_{t+1} \in [k]$ we have
\begin{equation*}
	\bH(V ~|~ \bvec{A}'_{\prec [\bvec{i}, i_{t+1}]}, W)
    	= \bH(V ~|~ \bvec{A}'_{[\bvec{i}, < i_{t+1}]}, \bvec{A}^{(1)}_{\prec \bvec{i}}, \bvec{A}^{(2)}_{\prec \bvec{i}},\ldots \bvec{A}^{(k)}_{\prec \bvec{i}}, W) \ .
\end{equation*}
\label{lem:pushing-back-cond}
\end{lemma}

\begin{proof}[Proof of Lemma~\ref{lem:matrix-implies-arikan}]
		Consider mixing matrix $M \in \F_q^{k\times k}$ satisfying $(\eta, b)$-exponential polarizing. By Theorem~1.10 in \cite{BGNRS}, we conclude that associated \Arikan\ martingale is locally polarizing. We will now show that it also satisfies the \emph{strong suction at the low end} condition of exponential local polarization.

		Let us consider independent samples $(\bvec{A}^{(1)}, \bvec{B}^{(1)}), \ldots (\bvec{A}^{(k)}, \bvec{B}^{(k)}) \sim D_{t-1}$, and pair $(\bvec{A'}, \bvec{B'})$ as in the Definition~\ref{def:arikan-martingale}, and moreover let us consider for some fixed $\bvec{i} \in [k]^{t-1}$. Take $h := \H(\bvec{A}^{(s)}_{\bvec{i}} \,|\, \bvec{A}^{(s)}_{\prec \bvec{i}}, \bvec{B})$ for any $s$ (this value does not depend on the choice of $s$).

We wish to show that for $i_t \sim \mathrm{Unif}([k])$, we have $\bH(\bvec{A'}_{ [\bvec{i}, i_t]} ~|~ \bvec{A'}_{\prec [\bvec{i}, i_t]}, \bvec{B'}) < h^b$ with probability at least $\eta$ over choice of random $i_t$. We can apply Lemma~\ref{lem:pushing-back-cond} to deduce
\begin{align*}
		\bH(\bvec{A'}_{ [\bvec{i}, i_t]} ~|~ \bvec{A'}_{\prec [\bvec{i}, i_t]}, \bvec{B'}) & = \bH(\bvec{A'}_{ [\bvec{i}, i_t]} ~|~ \bvec{A'}_{[\bvec{i}, <i_t]}, \bvec{A}^{(1)}_{\prec \bvec{i}}, \ldots, \bvec{A}^{(k)}_{\prec \bvec{i}}, \bvec{B'}) \\
		& = \bH( (\bvec{\tilde{A}} M)_{i_t} ~|~ (\bvec{\tilde{A}} M)_{<i_t}, \bvec{A}^{(1)}_{\prec \bvec{i}}, \ldots \bvec{A}^{(k)}_{\prec \bvec{i}}, \bvec{B'})
\end{align*}
where $\bvec{\tilde{A}} = (\bvec{A}^{(1)}_{\bvec{i}}, \ldots \bvec{A}^{(k)}_{\bvec{i}}) \in \F_q^k$, and the second identity follows from definition of $\bvec{A'}$. 

We can now apply the definition of exponential polarization of a matrix (Definition~\ref{def:matrix-exp-polar}), with $U_j = \tilde{A}_{\bvec{i}}^{(j)}$ and with $A_j := (\bvec{A}^{(j)}_{\prec \bvec{i}}, \bvec{B'})$ to conclude that for $\eta$ fraction of indiced $i_t$ this quanitity is bounded by $h^b$, as required.

\end{proof}

\section{Standard probabilistic inequalities \label{sec:appendix}}
\begin{lemma}[Lemma 2.2 in \cite{BGNRS}]
\label{lem:entropy-gives-prediction}
For a pair of random variables $(U_1, U_2) \in \Sigma_1 \times \Sigma_2$ there exists function $f : \Sigma_2 \to \Sigma_1$ such that $\P(f(U_2) \not= U_1) \leq H(U_1 | U_2)$.
\end{lemma}

\begin{lemma}[Fano's inequality]
\label{lem:prediction-gives-entropy}
For a pair of random variables $(U_1, U_2) \in \Sigma_1 \times \Sigma_2$, if we have a function $f : \Sigma_2 \to \Sigma_1$ such that $\P(f(U_2) \not= U_1) \leq \delta$ with $\delta < \frac{1}{2}$, then $H(U_2 | U_1) \leq 2 \delta (\log \delta^{-1} + \log \Sigma_1)$.
\end{lemma}

\begin{lemma}[Doobs martingale inequality]
	\label{lem:doobs}
	For any non-negative martingale $X$, we have
	\begin{equation*}
		\P(\sup_{t \leq T} X_t > \lambda) \leq \lambda^{-1} X_0
	\end{equation*}
\end{lemma}

We include the statements of following lemmas from \cite{BGNRS} for reference.
\jnote{Where to put those?}
\begin{lemma}
    \label{lem:sum-of-subexp}
    Consider a sequence of non-negative random variables $Y_1, Y_2, \ldots, Y_t, \ldots$ adapted to the sequence $X_t$. If for every $t$ we have $\Pr(Y_{t+1} > \lambda\,|\, X_{[1:t]}) < \exp(-\lambda)$, then for every $T>0$:
    %\arnote{Why do we need "almost surely" here?}, then
    % \todo[inline]{Jarek: $\Pr(Y_{t+1} > \lambda | X_{[1:t]})$ is a random variable, which depends on $X_{[1:t]}$, and we want it to be bounded as a random variable. I wanted to highlight that --- maybe it is better just to omit this phrase.}
\begin{equation*}
\P(\sum_{i \leq T} Y_i > CT) \leq \exp(-\Omega(T))
\end{equation*}
% \arnote{Better to say $\le \exp(-\Omega(T))$ than $=\exp(-\Omega(T))$ (even though the former is sort of implied by the $\Omega$ in the exponent). This change needs to be propagated. \checkm}
for some universal constant $C$.
\end{lemma}
\begin{lemma}
    \label{lem:sum-of-events}
    Consider a sequence of random variables $Y_1, Y_2, \ldots$ with $Y_i \in \{0, 1\}$, adapted to the sequence $X_t$. If $\P(Y_{t+1} = 1 | X_{[1:t]}) > \mu_{t+1}$ for some deterministic value $\mu_t$,  then for $\mu := \sum_{t\leq T} \mu_t$ we have
    \begin{equation*}
			\P(\sum_{t\leq T} Y_t < (1 - \varepsilon) \mu) \leq \exp(-\Omega(\varepsilon \mu))
    \end{equation*}
\end{lemma}

\end{document}